\documentclass[journal,nofonttune]{IEEEtran}

\usepackage{fontspec}
\usepackage{amsmath,amssymb,amsthm}
\usepackage{graphicx}
\usepackage{xcolor}
\usepackage{booktabs}
\usepackage{url}
\usepackage{braket}
\usepackage[numbers,sort&compress]{natbib}
\usepackage[margin=15mm]{geometry}
\usepackage{hyperref}
\usepackage{mathtools}
\usepackage[caption=false,font=footnotesize]{subfig}
\usepackage{algorithm}
\usepackage{algpseudocode}
\usepackage[most]{tcolorbox}
\usepackage{aliascnt}

\usepackage{todonotes}

\newcommand{\F}{\mathbb{F}}
\newcommand{\Fb}{\mathbb{F}_2}
\newcommand{\Fq}{\mathbb{F}_q}
\newcommand{\FQ}{\mathbb{F}_Q}

\newcommand{\wt}{\operatorname{wt}}
\newcommand{\pwc}{p_{\mathrm{wc}}}

\newcommand{\cB}{\mathcal{B}}
\newcommand{\cC}{\mathcal{C}}
\newcommand{\conv}{\operatorname{conv}}

\newcommand{\cin}{C^{\mathrm{I}}}
\newcommand{\nin}{n^{\mathrm{I}}}
\newcommand{\kin}{k^{\mathrm{I}}}
\newcommand{\din}{d^{\mathrm{I}}}
\newcommand{\rin}{R^{\mathrm{I}}}

\newcommand{\deltain}{\delta^{\mathrm{I}}}
\newcommand{\qin}{q}

\newcommand{\cout}{C^{\mathrm{O}}}
\newcommand{\nout}{n^{\mathrm{O}}}
\newcommand{\kout}{k^{\mathrm{O}}}
\newcommand{\dout}{d^{\mathrm{O}}}
\newcommand{\rout}{R^{\mathrm{O}}}
\newcommand{\routgs}{R^{\mathrm{O}}_{\mathrm{GS}}}
\newcommand{\deltaout}{\delta^{\mathrm{O}}}
\newcommand{\qout}{Q}
\newcommand{\doutlb}{\dout_{\mathrm{lb}}}

\newcommand{\ncon}{\nin\nout}
\newcommand{\kcon}{\kin\kout}
\newcommand{\dcon}{\din\dout}
\newcommand{\dconlb}{d^{\mathrm{C}}_{\mathrm{lb}}}
\newcommand{\deltacon}{\delta^{\mathrm{C}}}

\newcommand{\qcon}{\qin}

\newcommand{\ntarget}{n^*}
\newcommand{\ktarget}{k^*}

\newcommand{\nrac}{\nu}
\newcommand{\krac}{\kappa}
\newcommand{\rrac}{R_{\mathrm{RAC}}}

\newcommand{\rbase}{R_{\mathrm{B}}}

\newcommand{\np}{n'}
\newcommand{\kp}{k'}

\newtheorem{theorem}{Theorem}

\newaliascnt{corollary}{theorem}
\newtheorem{corollary}[corollary]{Corollary}
\aliascntresetthe{corollary}

\newaliascnt{lemma}{theorem}
\newtheorem{lemma}[lemma]{Lemma}
\aliascntresetthe{lemma}

\newaliascnt{conjecture}{theorem}
\newtheorem{conjecture}[conjecture]{Conjecture}
\aliascntresetthe{conjecture}

\newaliascnt{proposition}{theorem}
\newtheorem{proposition}[proposition]{Proposition}
\aliascntresetthe{proposition}

\theoremstyle{definition}

\newaliascnt{definition}{theorem}
\newtheorem{definition}[definition]{Definition}
\aliascntresetthe{definition}

\theoremstyle{definition}

\newaliascnt{remark}{theorem}
\newtheorem{remark}[remark]{Remark}
\aliascntresetthe{remark}

\usepackage[createShortEnv]{proof-at-the-end}
\newEndThm{remarkE}{remark}

\tcbset{
  apxthmstyle/.style={
    enhanced,
    breakable,
    colback=black!2,
    colframe=black!45,
    boxrule=0.4pt,
    arc=1pt,
    left=6pt,
    right=6pt,
    top=4pt,
    bottom=4pt,
    before skip=8pt,
    after skip=8pt,
  },
  apxproofstyle/.style={
    enhanced,
    breakable,
    colback=white,
    colframe=black!30,
    boxrule=0.4pt,
    arc=1pt,
    left=6pt,
    right=6pt,
    top=4pt,
    bottom=4pt,
    before skip=6pt,
    after skip=8pt,
  }
}

\let\oldtheorem\theorem
\let\endoldtheorem\endtheorem
\RenewDocumentEnvironment{theorem}{o}
  {\begin{tcolorbox}[apxthmstyle]%
   \IfNoValueTF{#1}{\oldtheorem}{\oldtheorem[#1]}}
  {\endoldtheorem\end{tcolorbox}}

\let\oldlemma\lemma
\let\endoldlemma\endlemma
\RenewDocumentEnvironment{lemma}{o}
  {\begin{tcolorbox}[apxthmstyle]%
   \IfNoValueTF{#1}{\oldlemma}{\oldlemma[#1]}}
  {\endoldlemma\end{tcolorbox}}

\let\oldcorollary\corollary
\let\endoldcorollary\endcorollary
\RenewDocumentEnvironment{corollary}{o}
  {\begin{tcolorbox}[apxthmstyle]%
   \IfNoValueTF{#1}{\oldcorollary}{\oldcorollary[#1]}}
  {\endoldcorollary\end{tcolorbox}}

\let\oldproposition\proposition
\let\endoldproposition\endproposition
\RenewDocumentEnvironment{proposition}{o}
  {\begin{tcolorbox}[apxthmstyle]%
   \IfNoValueTF{#1}{\oldproposition}{\oldproposition[#1]}}
  {\endoldproposition\end{tcolorbox}}

\let\olddefinition\definition
\let\endolddefinition\enddefinition
\RenewDocumentEnvironment{definition}{o}
  {\begin{tcolorbox}[apxthmstyle]%
   \IfNoValueTF{#1}{\olddefinition}{\olddefinition[#1]}}
  {\endolddefinition\end{tcolorbox}}

\let\oldremark\remark
\let\endoldremark\endremark
\RenewDocumentEnvironment{remark}{o}
  {\begin{tcolorbox}[apxthmstyle]%
   \IfNoValueTF{#1}{\oldremark}{\oldremark[#1]}}
  {\endoldremark\end{tcolorbox}}

\let\oldconjecture\conjecture
\let\endoldconjecture\endconjecture
\RenewDocumentEnvironment{conjecture}{o}
  {\begin{tcolorbox}[apxthmstyle]%
   \IfNoValueTF{#1}{\oldconjecture}{\oldconjecture[#1]}}
  {\endoldconjecture\end{tcolorbox}}
  
\title{Polynomial-Time Constructions of Asymptotically Good Random Access Codes from Linear Codes}
\author{
Ruho Kondo
and
Yota Maeda
\thanks{Ruho Kondo is with Toyota Central R\&D Labs., Inc., 41-1 Yokomichi, Nagakute, Aichi 480-1192, Japan (e-mail: r-kondo@mosk.tytlabs.co.jp).}
\thanks{Yota Maeda is with Toyota Central R\&D Labs., Inc., 1-4-14 Koraku, Bunkyo-ku, Tokyo 112-0004, Japan.}
}

\begin{document}
\maketitle

\begin{abstract}
Random access codes (RACs) encode a $q$-ary input string into fewer $q$-ary symbols while allowing any requested input symbol to be recovered with probability above the random-guessing value $1/q$.
However, it remains unclear what asymptotic guarantees on compression and recovery probability should be required of RACs.
We address this question by introducing the notion of \emph{asymptotically good RACs} and give two polynomial-time constructions.
First, an $[n,k,d]_q$ linear code with $d\ge2$, together with an encoder based on its orthogonal-array property, yields an $(n,n-k)_q$-RAC with worst-case decoding success probability
\[
\pwc=\frac{1}{q}+\left(1-\frac{1}{q}\right)\frac{d-2}{n}.
\]
These \emph{OA-RACs} are asymptotically good whenever the underlying linear codes are, and Garcia--Stichtenoth algebraic-geometry codes give a concrete polynomial-time realization.
Second, a direct-sum strategy using OA-RACs derived from $q$-ary Hamming codes and the trivial $(1,1)_q$-RAC yields asymptotically good RAC families at every prescribed RAC rate $\rrac=(n-k)/n\in(0,1)$.
Numerical results demonstrate favorable finite-length performance of both constructions.
\end{abstract}

\begin{IEEEkeywords}
random access codes, linear codes, dual distance, orthogonal arrays, algebraic geometry codes, concatenated codes
\end{IEEEkeywords}

\section{Introduction}
\label{sec:introduction}

A random access code (RAC) compresses an input string into a shorter message while allowing any designated input symbol to be decoded with probability strictly above random guessing.
For binary classical RACs, Ambainis et al.~\cite{ambainis1999dense} established an asymptotically tight tradeoff between the message rate and the decoding success probability, and the corresponding entropy converse extends to quantum messages~\cite{nayak1999optimal}.
Their achievability result, however, is nonconstructive, relying on existence arguments for suitable codes and a probabilistic encoder.
This motivates the study of explicit RAC constructions, particularly those admitting efficient encoding and decoding algorithms for practical implementations.

Explicit classical and quantum RACs have been investigated for various finite parameters and particular parameter families, with results on their constructions, decoding performance, and optimality~\cite{liabotro2017improved,ambainis2024quantum,tavakoli2015quantum,imamichi2018constructions,teramoto2023quantum,suzuki2026analytical,akibue2026measurement,kondo2026random,tan2026optimal}.
Recent work has also provided a geometric framework for explicit classical RAC constructions with general finite parameters~\cite{kondo2026random}, but no polynomial-time algorithm is known for the general construction.
Although direct sums of fixed-size RACs trivially yield polynomial-time constructible families with nonvanishing compression and worst-case decoding advantage, developing efficient explicit constructions that accommodate general parameter choices while maintaining these properties remains a challenge.

In this work, we introduce \emph{asymptotically good RACs}, whose
compression fractions and worst-case decoding advantages over random
guessing both remain bounded away from zero.
By the Ambainis--Nayak--Ta-Shma--Vazirani (ANTV) converse, these
conditions also keep the message rate bounded away from zero and the
worst-case decoding success probability bounded away from one.

We give two approaches to constructing such families.
First, building on the geometric connection between RACs and classical
codes~\cite{kondo2026random}, we introduce an OA encoder based on the
orthogonal-array property of linear codes.
Using dual codes as RAC codebooks, we show that any asymptotically good
family of linear codes yields an asymptotically good family of
\emph{OA-RACs}.
When the underlying codes and required operations are efficient, the
resulting OA-RACs are also efficient.
As a concrete polynomial-time realization, we use concatenated codes
with Garcia--Stichtenoth algebraic geometry outer codes
~\cite{garcia1995tower}.

Second, we consider a \emph{direct-sum strategy} that constructs large
RACs from smaller ones.
Using OA-RACs derived from $q$-ary Hamming codes together with the
trivial $(1,1)_q$-RAC, we show that this strategy constructs
asymptotically good RAC families at any prescribed rate.

Finally, we numerically examine the finite-length performance of both
constructions.
OA-RACs substantially outperform the Improved RAC construction
~\cite{liabotro2017improved} over much of the examined range, although
they do not generally attain the optimum.
For the direct-sum strategy, incorporating OA-RACs derived from
additional linear-code families, including algebraic geometry codes,
further improves the worst-case decoding success probability over a
wide range of RAC rates.
\section{Preliminaries}
\label{sec:preliminaries}
Our construction builds on the connection between RACs and linear codes established in~\cite{kondo2026random}.
We therefore first review RACs and the required facts on linear codes.
Since our explicit construction of asymptotically good RACs uses asymptotically good concatenated codes with algebraic geometry (AG) outer codes, we also review the necessary background on AG codes and concatenated codes.
Throughout the paper, for a positive integer $\ell$, we write $[\ell]:=\{0,1,\ldots,\ell-1\}$.
\subsection{Random Access Codes}
\label{subsec:rac-model}
An RAC with input length $\nrac$, message length $\krac$, and alphabet $\Fq$ is called an $(\nrac,\krac)_q$-RAC.
The encoder is a conditional distribution $P_{\mathrm{E}}(M\mid B)$ on messages $m\in\Fq^{\krac}$ given an input string $b\in\Fq^{\nrac}$.
For a requested coordinate $i\in[\nrac]$, the decoder is a conditional distribution $P_{\mathrm{D}}(\hat B_i\mid M)$ on the reconstructed symbol $\hat b_i\in\Fq$ given the message.
The worst-case decoding success probability is
\begin{align}
    \label{eq:rac-definition}
    \pwc(P_{\mathrm{E}},P_{\mathrm{D}})
    &:=\min_{b\in\Fq^{\nrac}}\min_{i\in[\nrac]}
    \sum_{m\in\Fq^{\krac}}
    P_{\mathrm{D}}(\hat B_i=b_i\mid M=m)\notag\\
    &\qquad\qquad\times P_{\mathrm{E}}(M=m\mid B=b).
\end{align}
The random-guessing probability is $1/q$.

For $q=2$, the maximum worst-case decoding success probability admits the geometric formulation~\cite{kondo2026random}
\begin{equation}
    \label{eq:worst-alt}
    \max_{P_{\mathrm{E}},P_{\mathrm{D}}}\pwc(P_{\mathrm{E}},P_{\mathrm{D}})
    =
    1-\min_{\substack{S\subset [0,1]^{\nrac}},\\|S|=2^{\krac}}
    \max_{b\in\Fb^{\nrac}}\min_{x\in\conv S} d_\infty(b,x),
\end{equation}
where $d_\infty(a,b):=\max_i|a_i-b_i|$ is the Chebyshev distance.
Thus, RAC construction can be viewed as a codebook-design problem over all subsets $S\subset[0,1]^{\nrac}$ of cardinality $2^{\krac}$.
Importantly, this search space contains the family of all ${\krac}$-dimensional binary linear codes $C\le\Fb^{\nrac}$, each of which also has exactly $2^{\krac}$ elements.
The present work exploits this inclusion by restricting $S$ to linear codebooks and then importing structural results from classical coding theory.

\subsection{Linear Codes}
\label{subsec:linear-codes}
For a prime power $q$, a subspace $C\le\Fq^n$ with $\dim C=k$ is called an $[n,k]_q$ linear code.
A matrix $G$ satisfying
\begin{equation}
    \label{eq:generator}
    C=\{uG:u\in\Fq^k\}
\end{equation}
is called a generator matrix of $C$.
The minimum distance is
\begin{equation}
    \label{eq:distance}
    d(C):=\min_{x\ne y\in C}d_{\rm H}(x,y)
    =\min_{0\ne c\in C}\wt(c).
\end{equation}
An $[n,k]_q$ code with minimum distance $d$ is denoted by $[n,k,d]_q$.
The dual code is
\begin{equation}
    \label{eq:dual}
    C^\perp=\{z\in\Fq^n:\langle z,c\rangle=0\ \text{for all }c\in C\},
\end{equation}
and $\dim C^\perp=n-k$.
We write $d^\perp:=d(C^\perp)$ when the underlying code is clear.
The rate and relative minimum distance are
\begin{equation}
    R:=\frac{k}{n},
    \qquad
    \delta:=\frac{d}{n}.
\end{equation}

Using these notations, asymptotically good linear codes are defined as follows:

\begin{definition}[Asymptotically good linear codes]
\label{def:asymptotically-good-code}
A sequence of $q$-ary linear codes $C_i=[n_i,k_i,d_i]_q$ with $n_i\to\infty$ is called asymptotically good if
\begin{equation}
    \liminf_{i\to\infty}\frac{k_i}{n_i}>0,
    \qquad
    \liminf_{i\to\infty}\frac{d_i}{n_i}>0.
\end{equation}
\end{definition}

\subsection{Algebraic Geometry Codes}
\label{subsec:ag-codes}
Let $X$ be a nonsingular projective irreducible curve of genus $g$ over $\FQ$.
For a divisor $G$, define the Riemann--Roch space
\begin{equation}
    \label{eq:riemann-roch-space}
    \mathcal L(G)=\{f\in\FQ(X)^\times:(f)+G\succeq0\}\cup\{0\}.
\end{equation}
Let $P_1,\ldots,P_n$ be distinct $\FQ$-rational points disjoint from the support of $G$, and set $D=P_1+\cdots+P_n$.
The corresponding evaluation code is
\begin{equation}
    \label{eq:ag-code}
    C_L(X,D,G)=\{(f(P_1),\ldots,f(P_n)):f\in\mathcal L(G)\}.
\end{equation}
If $\deg G<n$, the evaluation map is injective and
\begin{equation}
    K=\ell(G)\ge \deg G+1-g,
    \qquad
    D(C_L)\ge n-\deg G.
\end{equation}
Consequently,
\begin{equation}
    \label{eq:ag-bound}
    \delta\ge1-R-\frac{g-1}{n}.
\end{equation}
Thus, good asymptotic performance requires a large ratio of rational points to genus, $n/g$.

Let $n_Q(g)$ denote the maximum possible number of $\FQ$-rational points on a genus-$g$ curve over $\FQ$.
The Ihara constant is
\begin{equation}
    A(Q):=\limsup_{g\to\infty}\frac{n_Q(g)}{g}.
\end{equation}
The Drinfeld--Vl\u{a}du\c{t} bound gives $A(Q)\le\sqrt Q-1$, and equality is attained for square $Q$ by GS towers~\cite{garcia1995tower}.
Let $Q=r^2$ with $r$ a prime power.
Using the level index $m\in\{0,1,2,\ldots\}$ adopted here, one may select
\begin{equation}
    \label{eq:gs-length}
    n_m=Q^{m/2}(Q-\sqrt Q)
\end{equation}
rational evaluation points, while the genus is
\begin{equation}
    \label{eq:genus}
    g_m=
    \begin{cases}
    (Q^{(m+1)/4}-1)^2, & m\text{ odd},\\[1mm]
    (Q^{m/4}-1)(Q^{(m+2)/4}-1), & m\text{ even}.
    \end{cases}
\end{equation}
For a one-point divisor $G_m=\mu_mP_\infty$ with $\mu_m<n_m$, the AG code satisfies
\begin{equation}
    k_m\ge\mu_m+1-g_m,
    \qquad
    d_m\ge n_m-\mu_m.
\end{equation}
For a desired dimension $k_m\le n_m-g_m$, choosing $\mu_m=k_m+g_m-1$ yields a code of dimension at least $k_m$; a $k_m$-dimensional subcode then preserves the distance guarantee
\begin{equation}
    \label{eq:gs-parameters}
    d_m\ge n_m-k_m+1-g_m.
\end{equation}
The method of Shum et al.~\cite{shum2001low} constructs generator matrices of these GS-tower AG codes in $\widetilde O(n_m^3)$ field operations for fixed $Q$.
Moreover,
\begin{equation}
    \lim_{m\to\infty}\frac{g_m-1}{n_m}=\frac{1}{\sqrt Q-1},
\end{equation}
so the corresponding asymptotic distance guarantee is the Tsfasman--Vl\u{a}du\c{t}--Zink (TVZ) bound~\cite{tsfasman1982modular}
\begin{equation}
    \label{eq:tvz}
    \delta\ge1-R-\frac{1}{\sqrt Q-1}.
\end{equation}

A Reed--Solomon (RS) code over $\mathbb{F}_Q$ is defined by evaluating polynomials of degree less than $k$ at $n$ distinct elements $a_1,\ldots,a_n\in\mathbb{F}_Q$:
\begin{equation}
    \operatorname{RS}_Q(n,k)
    =\left\{(f(a_1),\ldots,f(a_n)):f\in\mathbb{F}_Q[x],\ \deg f<k\right\}.
\end{equation}
For $1\le k\le n\le Q$, this yields an $[n,k,n-k+1]_Q$ linear code.
In particular, an RS code is an AG code on the projective line $\mathbb{P}^1_{\mathbb{F}_Q}$, which has genus zero, with the divisor $G=(k-1)P_\infty$.
\subsection{Concatenated Codes}
\label{subsec:concat-codes}
Let the outer and inner linear codes have parameters $[\nout,\kout,\dout]_{\qout}$ and $[\nin,\kin,\din]_{\qin}$, respectively.
To map one outer symbol to one inner information vector, we impose
\begin{equation}
    \qout=\qin^{\kin}
\end{equation}
and fix an $\Fq$-linear vector-space isomorphism $\FQ\simeq\Fq^{\kin}$.
The concatenated code then has parameters
\begin{equation}
    [\ncon,\kcon,\ge\dcon]_{\qcon}.
\end{equation}

Let $[\nout_m,\kout_m,\dout_m]_\qout$ and $g_m$ be the parameters and genus of level-$m$ GS outer code.
For a fixed inner code concatenated with the level-$m$ GS outer code, the relative minimum distance satisfies
\begin{equation}
    \label{eq:gs-concat-distance}
    \deltacon_{m}
    \ge \deltain
    \left(1-\rout_{m}-\frac{g_m-1}{\nout_m}\right),
\end{equation}
where $\deltacon_{m}$ is the relative minimum distance of the concatenated code with level-$m$ GS outer code, and $\deltain=\din/\nin$, $\rout_{m}=\kout_m/\nout_m$.
Hence, if $\routgs:=\limsup_{m\to\infty} \rout_m$,
\begin{equation}
    \label{eq:gs-concat-asymptotic}
    \liminf_{m\to\infty}\deltacon_m
    \ge \deltain
    \left(1-\routgs-\frac{1}{\sqrt \qout-1}\right).
\end{equation}
In particular, the outer length can grow by increasing the GS tower level while the inner code and $\qout$ remain fixed.

For an RS outer code,
\begin{equation}
    \label{eq:rs-concat-distance}
    \deltacon_{\mathrm{RS}}
    \ge \deltain\left(1-\rout+\frac{1}{\nout}\right).
\end{equation}
However, an ordinary RS code satisfies $\nout\le \qout=\qin^{\kin}$.
Thus, at fixed base alphabet $\qin$, increasing the RS outer length eventually forces the inner dimension $\kin$ to grow as well.
In the explicit finite-length construction considered in this paper, the resulting certified inner-code relative-distance guarantee decreases with this growth, and consequently the certified relative minimum-distance guarantee of the RS concatenated construction approaches zero for very large target lengths.
For this reason, RS concatenation is used below as a finite-length baseline rather than as the mechanism for our asymptotically good RAC family.
The GS tower avoids this obstruction at the level of the guarantee because $\nout_m\to\infty$ can be achieved with a fixed inner code and fixed $\qout$.

\section{Construction of Asymptotically Good RACs}
\label{sec:dual-distance}

\subsection{Linear RACs}
\label{subsec:linear-rac}
Even after restricting~\eqref{eq:worst-alt} to $S\subset\{0,1\}^{\nrac}$, an optimal RAC need not in general be realized by a linear code~\cite{kondo2026random}.
Nevertheless, the family of linear codebooks is a particularly useful structured subclass because it exposes the full machinery of linear coding theory.
Moreover, in the binary worst-case formulation, restricting $S$ to $\{0,1\}^{\nrac}$ is equivalent to restricting the decoder to be deterministic~\cite{kondo2026random}.
Thus, after further restricting the codebook to a linear code, the remaining design choices are the linear code itself and the stochastic encoder used to select a codeword.
The transmitted message is simply an index of the selected codeword.
The resulting protocol has the form
\begin{equation}
\begin{array}{ccccc@{}c}
b
& \xrightarrow{\,P_{\mathrm{E}}(M|B=b)\,}
& m
& \xmapsto{\,\,}
& \tilde b={}
& c_m
\\[4pt]
\mathrel{\rotatebox[origin=c]{-90}{$\in$}}
& &
\mathrel{\rotatebox[origin=c]{-90}{$\in$}}
& &
&
\mathrel{\rotatebox[origin=c]{-90}{$\in$}}
\\[4pt]
\mathbb F_2^{\nrac}
& &
\mathbb F_2^{\krac}
& &
& C.
\end{array}
\end{equation}
To decode coordinate $i$, the receiver outputs $c_{m,i}$.

The discussion above is specific to the binary case.
We extend this framework to the $q$-ary setting and define a linear RAC as follows.
\begin{definition}[Linear RAC]
\label{def:linear-rac}
Fix a $\krac$-dimensional $q$-ary linear code $C\le\Fq^{\nrac}$ and a one-to-one indexing $m\in\Fq^{\krac}\leftrightarrow c_m\in C$.
Given $b\in\Fq^{\nrac}$, the encoder stochastically selects a codeword $c_m\in C$ and sends its index $m$.
For a requested coordinate $i$, the decoder outputs $c_{m,i}$.
We call this protocol a linear RAC.
\end{definition}

This restriction is the key bridge used in the rest of the paper: properties of linear codes can now be translated directly into guarantees for RACs.

\subsection{OA Encoder}
\label{subsec:oa-encoder}
An orthogonal array $\operatorname{OA}(M,n,v,\tau)$ is an $M\times n$ array over an alphabet of size $v$ such that, after selecting any $\tau$ columns, every possible $\tau$-tuple appears exactly $M/v^\tau$ times.
The integer $\tau$ is called the strength.
If the codewords of an $[n,k]_q$ linear code $C$ are arranged as the rows of an array, Delsarte's theorem~\cite{delsarte1973algebraic} gives
\begin{equation}
    \label{eq:delsarte-oa}
    C\text{ is an }\operatorname{OA}(q^k,n,q,\tau)
    \Longleftrightarrow
    d(C^\perp)\ge\tau+1.
\end{equation}
Using this property of linear codes, we construct the following RAC encoder:

\begin{theorem}[OA encoder]
\label{thm:oa-encoder}
Let $C=[n,k,d]_q$ be a linear code with dual distance $d^\perp\ge2$.
Choose any integer $s$ satisfying $0\le s\le d^\perp-2$.
Given $b\in\Fq^n$, choose a subset $J\subset[n]$ of size $s$ uniformly at random, then choose $c$ uniformly from
\begin{equation}
    \{c\in C:c_J=b_J\},
\end{equation}
and transmit the index of $c$.
For every $i\in[n]$ and every $b\in\Fq^n$, the decoding success probability of $(n,k)_q$-RAC is exactly
\begin{equation}
    \label{eq:oa-general}
    \frac{1}{q}+\left(1-\frac{1}{q}\right)\frac{s}{n}.
\end{equation}
In particular, choosing $s=d^\perp-2$ gives
\begin{equation}
    \label{eq:oa-finite}
    \frac{1}{q}+\left(1-\frac{1}{q}\right)\frac{d^\perp-2}{n}.
\end{equation}
\end{theorem}
\begin{proof}
By~\eqref{eq:delsarte-oa}, $C$ is an orthogonal array of strength at least $d^\perp-1$.
Because $|J|=s\le d^\perp-2$, every pattern on $J$ occurs equally often, so the conditional set $\{c\in C:c_J=b_J\}$ is nonempty and uniform sampling is well defined.
If $i\in J$, decoding is correct with probability $1$.
If $i\notin J$, then $|J\cup\{i\}|=s+1\le d^\perp-1$, so the projection onto $J\cup\{i\}$ is uniform; conditioned on $c_J=b_J$, the symbol $c_i$ is uniform on $\Fq$ and is correct with probability $1/q$.
Averaging over $J$ gives~\eqref{eq:oa-general}.
\end{proof}

We refer to an RAC incorporating this OA encoder as an \emph{OA-RAC}.
The theorem also allows a certified lower bound $d^\perp_{\rm lb}\le d^\perp$ and $d^\perp_{\rm lb}\ge2$ to be used in place of the exact dual distance as follows:

\begin{corollary}[Finite-length transfer from linear codes to RACs]
Let $\cB=[n,k,d]_q$ be a linear code, and suppose that a lower bound $d_{\mathrm{lb}}\le d$ with $d_{\mathrm{lb}}\ge 2$ is known.
Take $\cC:=\cB^\perp$ as the RAC codebook and apply the OA encoder of Theorem~\ref{thm:oa-encoder} with $s=d_{\mathrm{lb}}-2$.
Then the resulting $(n,n-k)_q$-RAC has worst-case decoding success probability
\begin{equation}
    \label{eq:col4}
    \frac{1}{q} + \left(1-\frac{1}{q}\right) \frac{d_{\mathrm{lb}}-2}{n}.
\end{equation}
\end{corollary}
\begin{proof}
Since $\cC=\cB^\perp$, $\dim \cC=n-k$ and $d(\cC^\perp)=d(\cB)=d\ge d_{\mathrm{lb}}$.
Hence $s=d_{\mathrm{lb}}-2$ satisfies the condition of Theorem~\ref{thm:oa-encoder}, and the claim follows directly from \eqref{eq:oa-general}. 
\end{proof}

This observation is important for polynomial-time construction because exact minimum-distance computation is intractable for general linear codes.

Beyond this computational advantage, the OA encoder also recovers RAC constructions that are optimal, either unconditionally or under the conjecture of~\cite{kondo2026random}, for certain parameter families.

\begin{corollary}[Known optimal binary RAC families]
\label{cor:known-optimal-families}
For $k\ge2$ and $n\ge2$, the OA encoder realizes both the worst-case-optimal $(2^k-1,k)_2$-RAC and the conditionally worst-case-optimal $(n,n-1)_2$-RAC identified in~\cite{kondo2026random}, where the latter optimality is conditional on the Conjecture 20 stated there.
\end{corollary}
\begin{proof}
The worst-case-optimal $(2^k-1,k)_2$-RAC of~\cite{kondo2026random} is realized by the binary simplex code $C_{\mathrm{simp}}$.
Its dual is the binary Hamming code, whose minimum distance is $3$.
Thus $d^\perp(C_{\mathrm{simp}})=3$, and the OA encoder uses $s=1$.
With $n=2^k-1$, Theorem~\ref{thm:oa-encoder} gives
\begin{equation}
    p_{\mathrm{wc}}
    =\frac{1}{2}+\frac{1}{2n},
\end{equation}
which is exactly the worst-case optimum reported in~\cite{kondo2026random}.

For the $(n,n-1)_2$ family, the relevant codebook is the binary single-parity-check code $C_{\mathrm{SPC}}$.
Its dual is the binary repetition code, whose minimum distance is $n$.
Hence $d^\perp(C_{\mathrm{SPC}})=n$, so the OA encoder uses $s=n-2$ and achieves
\begin{equation}
    p_{\mathrm{wc}}
    =\frac{1}{2}+\frac{1}{2}\frac{n-2}{n}
    =1-\frac{1}{n}.
\end{equation}
This is the success probability of the conditionally worst-case-optimal $(n,n-1)_2$-RAC in~\cite{kondo2026random}.
\end{proof}

\subsection{A Concrete Example of the OA Encoder}
\label{subsec:oa-example}
As a concrete binary example, consider a $[10,7]_2$ linear code $C$ whose dual code $C^\perp$ is a $[10,3]_2$ code with generator matrix
\begin{equation}
    G^\perp=H=
    \begin{bmatrix}
    1&1&0&0&0&0&1&1&0&1\\
    0&0&1&1&0&0&1&0&1&1\\
    0&0&0&0&1&1&0&1&1&1
    \end{bmatrix}.
\end{equation}
Equivalently, $H$ is a parity-check matrix of the primal code $C$.
Its RREF is already
\begin{equation}
    H_{\mathrm{RREF}}=
    \begin{bmatrix}
    1&1&0&0&0&0&1&1&0&1\\
    0&0&1&1&0&0&1&0&1&1\\
    0&0&0&0&1&1&0&1&1&1
    \end{bmatrix}.
\end{equation}
Using zero-based coordinate indices, the pivot columns are $(0,2,4)$.
After permuting columns so that the pivot columns appear first,
\begin{equation}
    \begin{bmatrix}
    1&0&0&1&0&0&1&1&0&1\\
    0&1&0&0&1&0&1&0&1&1\\
    0&0&1&0&0&1&0&1&1&1
    \end{bmatrix}
    =[I\mid A],
\end{equation}
where
\begin{equation}
    A=
    \begin{bmatrix}
    1&0&0&1&1&0&1\\
    0&1&0&1&0&1&1\\
    0&0&1&0&1&1&1
    \end{bmatrix}.
\end{equation}
Collect the pivot coordinates of a codeword into
\begin{equation}
    c_P=\begin{bmatrix}c_0&c_2&c_4\end{bmatrix}
\end{equation}
and the remaining free coordinates into
\begin{equation}
    c_F=\begin{bmatrix}c_1&c_3&c_5&c_6&c_7&c_8&c_9\end{bmatrix}.
\end{equation}
Then every codeword of $C$ is obtained by choosing $c_F\in\Fb^7$ and setting
\begin{equation}
    c_P=-c_FA^\mathsf{T}.
\end{equation}

The dual code generated by $G^\perp$ has minimum distance $5$, so $d^\perp(C)=5$ and Theorem~\ref{thm:oa-encoder} permits $s=3$.
Choose
\begin{equation}
    J=\{0,3,4\}.
\end{equation}
The codeword coordinates $0$, $3$, and $4$ must equal $b_0$, $b_3$, and $b_4$, respectively, which gives
\begin{align}
    \begin{bmatrix}
    c_0\\c_3\\c_4
    \end{bmatrix}
    &=\begin{bmatrix}
        c_1+c_6+c_7+c_9\\
        c_3\\
        c_5+c_7+c_8+c_9
    \end{bmatrix}
    =\begin{bmatrix}
    b_0\\b_3\\b_4
    \end{bmatrix}.
\end{align}
Simply fixing $c_3=b_3$ and sampling the remaining variables independently does not in general satisfy the two other constraints.
Instead, we sample uniformly from the affine solution space.

For concreteness, take
\begin{equation}
    (b_0,b_3,b_4)=(0,1,1).
\end{equation}
After fixing $c_3=1$, define
\begin{equation}
    x=\begin{bmatrix}c_1&c_5&c_6&c_7&c_8&c_9\end{bmatrix}
\end{equation}
and
\begin{equation}
    A'=\begin{bmatrix}
    1&0&1&1&0&1\\
    0&1&0&1&1&1
    \end{bmatrix},
    \qquad
    b'=\begin{bmatrix}0&1\end{bmatrix}.
\end{equation}
The remaining constraints are
\begin{equation}
    A'x^\mathsf{T}={b'}^\mathsf{T}.
\end{equation}
The augmented matrix
\begin{equation}
\label{eq:expand}
E=
\left[
    A'
    \,\Big|\,
    {b'}^\mathsf{T}
\right]
=\begin{bmatrix}
1&0&1&1&0&1&0\\
0&1&0&1&1&1&1
\end{bmatrix}
\end{equation}
is already in RREF.
The pivot columns among the variables are $\{0,1\}$ and the free columns are $\{2,3,4,5\}$.
Writing $x_P=(x_0,x_1)$ and $x_F=(x_2,x_3,x_4,x_5)$, one samples $x_F$ uniformly and then determines $x_P$ uniquely.
For example, take
\begin{equation}
    x_F=\begin{bmatrix}0&1&1&0\end{bmatrix}.
\end{equation}
Then
\begin{equation}
\begin{bmatrix}
    x_0+1\\
    x_1
\end{bmatrix}
=
\begin{bmatrix}
    0\\
    1
\end{bmatrix},
\end{equation}
so $x_P=(1,1)$ and
\begin{equation}
    x=\begin{bmatrix}1&1&0&1&1&0\end{bmatrix}.
\end{equation}
Together with $c_3=1$, this gives
\begin{align}
    c_P
    &=-c_FA^\mathsf{T}=c_FA^\mathsf{T}\notag\\
    &=\begin{bmatrix}
        1&1&1&0&1&1&0
    \end{bmatrix}
    \begin{bmatrix}
        1&0&0\\
        0&1&0\\
        0&0&1\\
        1&1&0\\
        1&0&1\\
        0&1&1\\
        1&1&1
    \end{bmatrix}
    =\begin{bmatrix}0&0&1\end{bmatrix},
\end{align}
where $-c_FA^\mathsf{T}=c_FA^\mathsf{T}$ over $\Fb$.
The final codeword is
\begin{equation}
    c=\begin{bmatrix}0&1&0&1&1&1&0&1&1&0\end{bmatrix}.
\end{equation}
Because every assignment of the free variables determines exactly one solution, uniform sampling of the free variables gives a uniform sample from $\{c\in C:c_J=b_J\}$.

\subsection{Asymptotically good RACs}
\label{subsec:asymp-rac}

For linear codes, asymptotic goodness requires both the code rate and the relative minimum distance to remain nonvanishing as the blocklength grows.
For RACs, we instead consider the compression fraction $1-\krac/\nrac$ and the decoding advantage $\pwc-1/q$.
This motivates the following definition.

\begin{definition}[Asymptotically good RACs]
\label{def:asymptotically-good-rac}
A sequence of $q$-ary $(\nrac_i,\krac_i)$-RACs with $\nrac_i\to\infty$ is called asymptotically good if
\begin{equation}
    \liminf_{i\to\infty}
    \left(1-\frac{\krac_i}{\nrac_i}\right)>0
\end{equation}
and
\begin{equation}
    \liminf_{i\to\infty}
    \left(\pwc{}_{,i}-\frac{1}{q}\right)>0.
\end{equation}
\end{definition}

By the ANTV converse, these two conditions also imply that the RAC rate $\rrac:=\krac/\nrac$ is bounded away from zero and that the decoding success probability is bounded away from one, respectively.

For example, consider the known infinite family of $(2^{\krac}-1,\krac)_2$-RACs~\cite{kondo2026random}.
Its compression fraction satisfies $1-\frac{\krac}{\nrac}=1-\frac{\krac}{2^{\krac}-1}\to 1$, whereas its worst-case decoding success probability is $p_{\mathrm{wc}}=\frac{1}{2}+\tfrac{1}{2(2^{\krac}-1)}$.
Hence, $\pwc-\frac{1}{2}\to0$, so this family is not asymptotically good.
As another example, consider the infinite family of $(\nrac,\nrac-1)_2$-RACs~\cite{kondo2026random}.
Its compression fraction is $1-\frac{\krac}{\nrac}=\frac{1}{\nrac}\to0$, while
$\pwc=\frac{\nrac-1}{\nrac}\to1$.
Thus, this family is also not asymptotically good, because its compression fraction vanishes.
These examples illustrate the two defining requirements separately.
Moreover, the ANTV converse implies that a nonvanishing decoding advantage automatically keeps the RAC rate bounded away from zero, while a nonvanishing compression fraction automatically keeps the decoding success probability bounded away from one.

Definition~\ref{def:asymptotically-good-rac} does not prescribe the asymptotic RAC rate.
To distinguish asymptotically good families by their limiting rate, we introduce the following refinement.

\begin{definition}[Asymptotically good RACs at rate $\rrac$]
\label{def:asymptotically-good-rac-R}
Let $\rrac\in(0,1)$.
An asymptotically good sequence of $(\nu_i,\kappa_i)_\qin$-RACs is said to be asymptotically good at rate $\rrac$ if
\begin{equation}
    \lim_{i\to\infty}\frac{\kappa_i}{\nu_i}=\rrac.
\end{equation}
\end{definition}

We next extend this notion from a single prescribed rate to all rates in $(0,1)$ by defining the following property of a class of RACs.

\begin{definition}[Rate-universal asymptotic goodness]
\label{def:rate-universal}
A class $\mathcal{F}$ of $q$-ary RACs is called rate-universally asymptotically good if, for every $\rrac\in(0,1)$, $\mathcal{F}$ contains a sequence of RACs that is asymptotically good at rate $\rrac$.
\end{definition}

The next theorem states the general transfer principle from coding theory to RACs.
\begin{theorem}[Asymptotically good codes induce asymptotically good RACs]
\label{thm:good-code-to-good-rac}
Let $\{B_i\}$ be an asymptotically good family of $q$-ary linear codes
\begin{equation}
    B_i=[n_i,k_i,d_i]_q,
    \qquad n_i\to\infty.
\end{equation}
For each $i$, take the RAC codebook to be
\begin{equation}
    C_i:=B_i^\perp
\end{equation}
and use the OA encoder with $s_i=d_i-2$ for all sufficiently large $i$.
Then the resulting $(n_i,n_i-k_i)_q$-RAC family is asymptotically good.
\end{theorem}

\begin{proof}
Because $\{B_i\}$ is asymptotically good, $\delta_*:=\liminf_{i\to\infty}d_i/n_i>0$.
The Singleton bound gives $n_i-k_i\ge d_i-1$, and therefore $\liminf_{i\to\infty}(n_i-k_i)/n_i\ge \delta_*>0$.
Moreover, since $\{B_i\}$ is asymptotically good, $\limsup_{i\to\infty}(n_i-k_i)/n_i=1-\liminf_{i\to\infty}k_i/n_i<1$.
Thus the RAC message rate is bounded away from both zero and one.
Moreover, $d(C_i^\perp)=d(B_i)=d_i$.
By Theorem~\ref{thm:oa-encoder}, $p_{\mathrm{wc},i}-1/q = \left(1-1/q\right)(d_i-2)/n_i$.
Taking the lower limit gives
$\liminf_{i\to\infty}\left(p_{\mathrm{wc},i}-1/q\right)\ge\left(1-1/q\right)\delta_*>0$.
Hence the RAC family is asymptotically good.
\end{proof}

\subsection{Code Design}
\label{subsec:code-design}
Theorem~\ref{thm:good-code-to-good-rac} reduces the construction of asymptotically good RACs to that of polynomial-time constructible asymptotically good linear-code families with efficiently available distance certificates.
Thus, any explicit code family satisfying these requirements, such as Justesen-type codes, can in principle be used in our construction.
Random linear codes, for example, attain the GV bound with high probability, but computing the exact minimum distance of a general binary linear code is NP-hard~\cite{vardy1997intractability}.
A randomly generated code therefore does not by itself provide the deterministic, efficiently available distance certificate required here.

In this work, we use concatenated codes with GS-tower AG codes as outer codes.
This construction allows the inner code and the outer alphabet to remain fixed as the blocklength grows, while providing a convenient certified distance bound for finite-length analysis.
We assume that the inner code is selected from a database of known linear codes for implementability.

For a GS-tower outer code, the concatenation condition $\qout=\qin^{\kin}$ must be compatible with the requirement that $\qout$ be a square.
Thus either the inner dimension $\kin$ is even or the base alphabet $\qin$ is itself a square.
For binary RACs, $\qin=2$, so the inner dimension must be even.
We assume that the inner code is selected from a database of known linear codes.
The GS construction provides an explicit asymptotically good code family
to which Theorem~\ref{thm:good-code-to-good-rac} can be applied.

\begin{corollary}[GS-tower construction of asymptotically good RACs]
\label{cor:gs-good-rac}
Let $\cin=[\nin,\kin,\din]_\qin$ be a fixed inner code with $\kin,\din>0$, and suppose that $\qout=\qin^{\kin}>4$ is a square.
Fix $0<\rout<1-1/(\sqrt{\qout}-1)$.
For each GS-tower level $m$, choose an outer AG code of dimension $\kout_m=\lfloor \rout \nout_m\rfloor$ satisfying $\dout_m\ge \nout_m-\kout_m+1-g_m$.
Then the resulting concatenated codes $\cB_m$ form an asymptotically good family of $\qin$-ary linear codes.
Consequently, the dual codebooks $\cC_m=\cB_m^\perp$, together with the OA encoder, form an asymptotically good RACs.
\end{corollary}

\begin{proof}
The GS tower yields an asymptotically good family of AG codes over $\F_Q$~\cite{garcia1995tower}. Since the inner code is fixed and has positive rate and relative minimum distance, the standard concatenated-code rate and distance bounds imply that ${\cB_m}$ is asymptotically good.
The claim then follows from Theorem~\ref{thm:good-code-to-good-rac}.
\end{proof}

Throughout this paper, we refer to RACs constructed using GS-tower AG concatenated codes as base codes and equipped with the OA encoder as \emph{AG/OA-RACs}.

By contrast, the RS construction used in this paper is primarily a finite-length benchmark.
Its certified distance lower bound is often strong at moderate lengths, but the RS length constraint $\nout\le \qout$ forces the inner dimension to grow when the target concatenated length becomes very large.
Under the explicit inner-code selection procedure used here, this causes the certified relative minimum-distance guarantee to decrease toward zero.
The numerical results in Section~\ref{sec:numerics} illustrate this behavior directly.
We refer to these RACs as \emph{RS/OA-RACs}.

AG and RS concatenated codes constrain the achievable length and dimension.
To obtain a desired $[n,k]$ linear code, we combine shortening and puncturing while preserving the certified minimum-distance lower bound as much as possible.
Puncturing deletes one coordinate.
If a nonpivot column of a systematic generator matrix is punctured, an $[n,k,d]$ code becomes an $[n-1,k,\ge d-1]$ code without reducing dimension.
Shortening at a coordinate on which the code is not identically zero yields an $[n-1,k-1,\ge d]$ code.
Thus, we first shorten until the desired dimension is reached, choosing coordinates that reduce the dimension by one at each step.
If the resulting length is larger than desired, we puncture nonpivot coordinates; if it is smaller, we append zero coordinates.
If an $[\np,\kp,d']$ code is transformed into an $[\ntarget,\ktarget,d]$ code in this way, then
\begin{equation}
    \label{eq:shorten-puncture-distance}
    d\ge d'-\max\{0,\np-\ntarget-\kp+\ktarget\}.
\end{equation}

\algrenewcommand\algorithmicrequire{\textbf{Input:}}
\algrenewcommand\algorithmicensure{\textbf{Output:}}
\begin{algorithm}[t]
\caption{Adjustment of code length and dimension by shortening and puncturing. $\texttt{Shortening}(G,x)$ applies $x$ shortening operations at coordinates that reduce the dimension by one, $\texttt{Puncturing}(G,x)$ applies $x$ dimension-preserving punctures, and $\texttt{ZeroPadding}(G,x)$ appends $x$ zero coordinates.}
\label{alg:shortening-puncturing}
\begin{algorithmic}[1]
\Require $G,\ntarget,\ktarget,\np,\kp$
\Ensure Modified generator matrix
\State $G\gets\texttt{Shortening}(G,\kp-\ktarget)$
\State $n_s\gets\np-(\kp-\ktarget)$
\If{$n_s-\ntarget>0$}
    \State $G\gets\texttt{Puncturing}(G,n_s-\ntarget)$
\Else
    \State $G\gets\texttt{ZeroPadding}(G,\ntarget-n_s)$
\EndIf
\State \Return $G$
\end{algorithmic}
\end{algorithm}

\subsection{Asymptotic Performance of the AG/OA-RACs with Fixed Inner-Code Database}
\label{eq:prac-asymp-racs}
\begin{figure}[tb]
    \centering
    \includegraphics[width=.9\columnwidth]{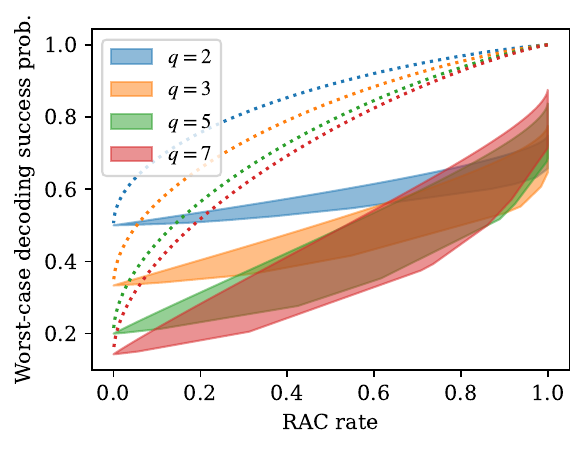}
    \caption{
    Theoretical performance envelope for the asymptotic worst-case decoding success probability of AG/OA-RACs, derived using the finite-length linear GV lower bound for the inner codes with $\nin\le 30$, the TVZ lower bound for the outer codes, and the MRRW upper bound for the concatenated codes.
    The actual asymptotic performance curve with $\nin\le30$ lies within the shaded region.
    The ANTV upper bounds are shown for reference as dashed-line.
    }
    \label{fig:potential}
\end{figure}

As discussed above, AG/OA-RACs are asymptotically good, whereas RS/OA-RACs are not when the inner codes are restricted to those in the fixed database.
In what follows, we describe the asymptotic performance of the AG/OA-RACs with finite-length inner codes from the viewpoint of practical implementability.

Let the inner code be $\cin$ with parameters $[\nin, \kin, \din]_\qin$, $\rin := \kin/\nin$, $\deltain := \din/\nin$, and let the outer code be $\cout$ with parameters $[\nout, \kout, \dout]_\qout$, $\rout := \kout/\nout$, $\deltaout := \dout/\nout$, and let $R_{\rm RAC}:=\krac/\nrac$ and $R_{\rm B}:=1-R_{\rm RAC}$ denote the $(\nrac,\krac)_\qin$-RAC rate and the rate of the concatenated base code, respectively.

Since the minimum distance of a concatenated code is at least $\dout\din$, we use this value as $d_{\mathrm{lb}}$ in \eqref{eq:col4} to evaluate the worst-case decoding success probability.
Accordingly, for fixed inner-code parameters $[\nin,\kin,\din]$, the asymptotic worst-case decoding success probability of the AG/OA-RAC is given by
\begin{align}
    \label{eq:asymptotic-certificate}
    &p_\qin(R_{\mathrm{B}};\nin,\kin,\din)
    = \frac{1}{\qin}
    + \left(1-\frac{1}{\qin}\right)\notag\\
    &\qquad\times
    \liminf_{i\to\infty}\frac{\dout_i}{\nout_i}
    \frac{\din}{\nin}.
\end{align}
Since the inner code is designed to be retrieved from a database, $\nin$ is finite.
Thus, the maximum achievable minimum distance of the inner code $\din_{\mathrm{max}}$ for specified $\nin$, $\kin$, and $\qin$ satisfies $\din_{\mathrm{max}}(\nin,\kin;\qin) \ge d_{\mathrm{GV}}(\nin,\kin;\qin)$, where $d_{\mathrm{GV}}(\nin,\kin;\qin)$ denotes the finite-length GV bound.
On the other hand, one advantage of AG concatenated codes is that their asymptotic performance can be evaluated even when the inner code is fixed.
Hence, the outer-code length can be taken to infinity, i.e., $\nout \to \infty$.
Consequently, the maximum achievable minimum distance of the outer code $\dout_{\mathrm{max}}$ for specified $\nout$, $\kout$, and $\qout$ satisfies $d^{\mathrm{O}}_{\mathrm{max}}(\nout,\kout;\qout)\ge d_{\mathrm{TVZ}}(\nout,\kout;\qout)$, where $d_{\mathrm{TVZ}}(\nout,\kout;\qout)$ denotes the TVZ bound in \eqref{eq:tvz}.
Note that $\kin$ must be even whenever $\qin$ is not a square, since $\qout=\qin^{\kin}$ must be a square.
An upper bound on the worst-case decoding success probability can also be obtained by applying the MRRW bound~\cite{mceliece1977new} to the relative minimum distance of the concatenated code.

We set $\qin\in\{2,3,5,7\}$ and $\nin\le30$.
Then, the numerical calculation is performed as a function of $\rbase$ and plotted against
$\rrac=1-\rbase$.
Figure~\ref{fig:potential} shows the obtained bounds of worst-case decoding success probability of AG/OA-RAC, together with the ANTV upper bound, $\pwc\le1-\mathcal{H}^{-1}_\qin(1-\rrac)$ where $\mathcal{H}_\qin^{-1}$ denotes the inverse of the $\qin$-ary entropy function on $[0,1-1/\qin]$.
This figure shows that, although the AG/OA-RAC does not attain the ANTV bound, its worst-case decoding success probability exceeds the random-guessing probability.
Note that these asymptotic bounds allow the concatenated-code length to vary freely.
If the concatenated-code length is instead fixed across all values of $\rbase$, the resulting constraints on the outer code generally prevent the achievable performance from reaching this asymptotic region.
We illustrate this effect in detail in Section~\ref{sec:numerics}.

\subsection{Computational Complexity}
\label{subsec:protocol-complexity}

We evaluate the theoretical computational costs of the five modules involved
in the construction of the proposed RACs.
The target is an $(\nrac,\krac)_{\qin}$-RAC.
Throughout this subsection, each arithmetic operation over a fixed finite
field is treated as having unit cost.

\emph{Parameter identification.}

Before constructing the concatenated codes, we first determine the parameters of the inner and outer codes.
In the finite-length implementation considered in this work, the inner code is selected from a fixed database $\mathcal D$.
The outer-code parameters are then chosen so that the resulting concatenated code matches the target length and dimension.
Hence, for a fixed database, the parameter-identification cost is
\begin{equation}
    T_{\mathrm{identify}}=O(1)
\end{equation}
with respect to the target blocklength over the range of
parameters supported by $\mathcal D$.

\emph{Concatenated-code construction.}
For fixed $\qout$, a generator matrix of the GS-tower AG outer code can be constructed using $\widetilde O(\nout{}^3)$ field operations by the method of \cite{shum2001low}.
On the other hand, an RS outer code can be constructed using $O(\nout\kout)$ field operations.
In addition, expanding the outer-code generator matrix over the base field and composing it with the inner code requires, using ordinary matrix multiplication, $O(\kout\nout\kin{}^2\nin)$ field operations.
Thus, the complexities of this module are
\begin{align}
    T_{\mathrm{concat}}^{\mathrm{AG}}
    &=\widetilde O(\nout{}^3)
      +O(\kout\nout\kin{}^2\nin),\\
    \label{eq:T-RS-concat}
    T_{\mathrm{concat}}^{\mathrm{RS}}
    &=O(\nout\kout+\kout\nout\kin{}^2\nin).
\end{align}
For the GS-tower AG construction, the inner code and hence $\qout=\qin^{\kin}$ are fixed while the tower level is allowed to grow.
If the outer-code rate is bounded away from zero and one, then $\nout=\Theta(\nrac)$ and $K=\Theta(\nout)$.
Consequently,
\begin{equation}
    T_{\mathrm{concat}}^{\mathrm{AG}}
    =
    \widetilde{O}(\nrac^3).
\end{equation}
For the RS construction, Eq.~\eqref{eq:T-RS-concat} gives the arithmetic complexity for each feasible finite-length parameter set.
However, an ordinary RS outer code must satisfy $\nout \le \qout = \qin^{\kin}$.
Hence, with a fixed base alphabet and a fixed inner code, the outer length $\nout$ cannot grow without bound.
More generally, if the inner-code database $\mathcal{D}$ is fixed and finite, the set of reachable RS blocklengths is also bounded.
Therefore, under the finite-length RS design considered in this work, we do not interpret Eq.~\eqref{eq:T-RS-concat} as an asymptotic complexity bound in $\nrac$.
The RS construction is used only as a finite-length benchmark.

\emph{Code-size modification.}
Let the generator matrix before shortening and puncturing have size
$\kp\times\np$.
To puncture coordinates without reducing the dimension, nonpivot
coordinates must be identified, for which we first reduce the generator
matrix to RREF.
Using ordinary Gaussian elimination, this requires
$O(\kp^2\np)$ field operations.
Once the pivot structure is known, shortening, puncturing, and zero padding
can be performed by row and column operations requiring at most
$O(\kp\np)$ additional operations.
Therefore,
\begin{equation}
    \label{eq:complexity-modify}
    T_{\mathrm{modify}}
    =O(\kp^2\np).
\end{equation}
If both $\kp$ and $\np$ grow linearly with the target RAC length, this
becomes
\begin{equation}
    T_{\mathrm{modify}}=O(\nrac^3).
\end{equation}

\emph{OA-sampler construction.}
The dominant one-time preprocessing step is reducing an $(\nrac-\krac)\times \nrac$ parity-check matrix $G^\perp=H$ of the RAC codebook to RREF.
With ordinary Gaussian elimination, this requires
\begin{equation}
    \label{eq:T-sampler}
    T_{\mathrm{sampler}}
    =
    O\!\left((\nrac-\krac)^2\nrac\right)
\end{equation}
field operations over $\Fq$.

\emph{OA sampling.}
The computational cost of OA sampling is dominated by the reduction of the augmented matrix $E$ in~\eqref{eq:expand}, which has at most $s$ rows and $\krac+1$ columns, to RREF.
This requires
\begin{equation}
    T_{\mathrm{sample}}
    =
    O\!\left(s^2\krac\right)
\end{equation}
field operations.
If the RAC rate is bounded away from zero and the relative dual distance is bounded away from zero, then $\krac=\Theta(\nrac)$ and $s=\Theta(\nrac)$.
Thus, the worst-case complexity of this module is
\begin{equation}
    \label{eq:T-sample}
    T_{\mathrm{sample}}=O(\nrac^3).
\end{equation}

Combining these estimates, the overall worst-case complexity of the AG/OA-RAC construction along the GS-tower family considered in this work is
\begin{equation}
    \widetilde{O}(\nrac^3).
\end{equation}
Therefore, the explicit AG/OA-RAC family can be constructed and implemented in polynomial time as a function of the target blocklength.

\subsection{Direct-Sum Strategy}
\label{sec:directsum}

In this section, we consider a strategy that partitions both the input information and the message into multiple blocks.
We refer to such an RAC as a direct-sum RAC (DS-RAC) and formally define it as follows.

\begin{definition}[Direct-sum RAC (DS-RAC)]
We denote a $q$-ary $(\nrac,\krac)$-RAC with worst-case decoding success probability $\pwc$ by an $(\nrac,\krac,\pwc)_q$-RAC, and let $\mathcal{R}$ denote a specific RAC protocol with these parameters.
Let
\begin{equation}
    \mathcal{D}_{\mathrm{RAC}}
    =
    \left\{
    \mathcal{R}_j :
    \mathcal{R}_j
    \text{ is a }
    (\nu_j,\kappa_j,\pwc{}_{,j})_q\text{-RAC}
    \right\}_j
\end{equation}
be a dataset of available $q$-ary RAC protocols.
A direct-sum RAC (DS-RAC) is an RAC constructed by taking the direct sum of constituent RACs selected from $\mathcal{D}_{\mathrm{RAC}}$, according to the following procedure.
\begin{itemize}
    \item \textbf{Block selection:}
    Select $L$ constituent RACs $(\nrac^{(\ell)},\krac^{(\ell)},\pwc{}^{(\ell)})_q$, $\ell=1,\ldots,L$, from $\mathcal{D}_{\mathrm{RAC}}$, allowing repetitions, such that
    \begin{equation}
        \sum_{\ell=1}^{L}\nrac^{(\ell)}=\nrac,
        \qquad
        \sum_{\ell=1}^{L}\krac^{(\ell)}=\krac.
    \end{equation}
    \item \textbf{Encoding:}
    Partition the input $b\in\mathbb{F}_q^{\nrac}$ into $L$ blocks $b^{(1)},\ldots,b^{(L)}$ of lengths $\nrac^{(1)},\ldots,\nrac^{(L)}$, respectively.
    Independently encode each block using its corresponding constituent RAC to obtain a message of length $\krac^{(\ell)}$.
    Concatenate the resulting messages and transmit the concatenated message.
    \item \textbf{Decoding:}
    Given a requested coordinate $i\in[\nrac]$, identify the block $\ell$ containing $i$ and recover the corresponding input symbol using the decoder of the $\ell$-th constituent RAC.
\end{itemize}
\end{definition}

Maximizing the worst-case decoding success probability of this protocol amounts to solving
\begin{align}
    &\max_{\substack{\{\nrac^{(\ell)},\krac^{(\ell)}\}_\ell}}
    \min_\ell p^{(\ell)}_{\mathrm{wc}}
    \notag\\
    &\text{s.t.}\quad
    (\nrac^{(\ell)},\krac^{(\ell)},p_{\mathrm{wc}}^{(\ell)})\in\mathcal{D}_{\mathrm{RAC}},\notag\\
    &\sum_\ell \nrac^{(\ell)}=\nrac,
    \qquad
    \sum_\ell \krac^{(\ell)}=\krac.
\end{align}
For a fixed dataset $\mathcal{D}_{\mathrm{RAC}}$, the above optimization problem can be reduced to a knapsack-type problem.
Thus, an optimal solution can be obtained by dynamic programming with computational complexity 
\begin{equation}
    T_{\mathrm{DS}}=O\left(\nrac\krac|\mathcal{D}_{\mathrm{RAC}}|\right).
\end{equation}
More specifically, an optimal partition can be found using the procedure described in Algorithm~\ref{alg:directsum}.

We now show that the direct-sum strategy yields rate-universally asymptotically good RACs through a simple construction.

\begin{theorem}[Rate-universal asymptotic goodness of the direct-sum construction]
\label{theorem:directsum}
The direct-sum strategy with the dataset of $q$-ary Hamming/OA-RACs and the trivial $(1,1)_q$-RAC is rate-universally asymptotically good.
\end{theorem}
\begin{proof}
Recall that $q$-ary Hamming codes have parameters $[n,k,d]=\left[\frac{q^m-1}{q-1},\frac{q^m-1}{q-1}-m,3\right]$ for $m\ge2$.
Hence, the parameters of the corresponding Hamming/OA-RACs are
\begin{equation}
    \nrac_m=\frac{q^m-1}{q-1},
    \quad
    \krac_m=m,
    \quad
    \pwc{}_{,m}
    =
    \frac{1}{q}
    +
    \frac{(q-1)^2}{q(q^m-1)}.
\end{equation}
Let
\begin{equation}
    r_m
    :=
    \frac{\krac_m}{\nrac_m}
    =
    \frac{m(q-1)}{q^m-1},
    \qquad m\ge2,
\end{equation}
and set $(\nrac_1,\krac_1,r_1,\pwc{}_{,1})=(1,1,1,1)$, corresponding to the trivial $(1,1)_q$-RAC.
Then $r_m$ decreases to $0$ as $m\to\infty$.
Hence, for every $\rrac\in(0,1)$, there exists $m\ge2$ such that
\begin{equation}
    r_m\le\rrac\le r_{m-1}.
\end{equation}
By taking direct sums of the $m$-th and $(m-1)$-st constituent RACs in suitable proportions, we can construct a sequence whose blocklength tends to infinity and whose rate converges to $\rrac$.
Indeed, the rate of a direct sum consisting of $a$ copies of the $m$-th RAC and $b$ copies of the $(m-1)$-st RAC is $(a\krac_m+b\krac_{m-1})/(a\nrac_m+b\nrac_{m-1})$, and these rates can be made arbitrarily close to any value in $[r_m,r_{m-1}]$ by varying $a/b$.
Since $\pwc{}_{,m}<\pwc{}_{,m-1}$ for $m\ge2$, where $\pwc{}_{,1}=1$, each such direct sum has worst-case decoding success probability $\pwc{}_{,m}$.
Choose a sequence $\{(\nrac_i,\krac_i,\pwc{}_{,i})\}_{i=1}^{\infty}$ of such direct sums so that
\begin{equation}
    \frac{\krac_i}{\nrac_i}\to\rrac
    \quad\text{and}\quad
    \pwc{}_{,i}=\pwc{}_{,m}
    \quad\text{for all } i.
\end{equation}
Since $\pwc{}_{,m}>\frac{1}{q}$, we have
\begin{equation}
    \liminf_{i\to\infty}
    \left(
        \pwc{}_{,i}-\frac{1}{q}
    \right)
    =
    \pwc{}_{,m}-\frac{1}{q}
    =
    \frac{(q-1)^2}{q(q^m-1)}
    >0.
\end{equation}
Moreover,
\begin{equation}
    \liminf_{i\to\infty}
    \left(
        1-\frac{\krac_i}{\nrac_i}
    \right)
    =
    1-\rrac
    >0.
\end{equation}
Thus, this sequence is asymptotically good at rate $\rrac$.
Since $\rrac\in(0,1)$ was arbitrary, the DS-RAC with a dataset containing all $q$-ary Hamming/OA-RACs and the trivial $(1,1)_q$-RAC is rate-universally asymptotically good.
\end{proof}

Finally, although it suffices for the dataset to contain Hamming/OA-RACs in order to construct rate-universally asymptotically good RACs via the direct-sum strategy, incorporating other RACs into the dataset can further improve the worst-case decoding success probability. We investigate this improvement in detail through numerical experiments in a later section.

\begin{algorithm}[tb]
\caption{
Pseudocode for dynamic programming to find the optimal $(\nrac, \krac)$ partition $\mathcal{P}$ and the corresponding worst-case decoding success probability $p_{\mathrm{wc}}^\oplus$ under the direct-sum strategy.
$\|$ denotes list concatenation.
}
\label{alg:directsum}
\begin{algorithmic}[1]
\Require $\nrac$, $\krac$, $q$, $\mathcal{D}_{\mathrm{RAC}}$
\Ensure $p_{\mathrm{wc}}^\oplus$, $\mathcal{P}$
\State $\mathrm{DP}_{0,0} \gets \infty$
\For{$a=0,1,\dots,\nrac$}
    \For{$b=0,1,\dots,\krac$}
        \If{$(a,b)=(0,0)$}
            \State \textbf{continue}
        \EndIf
        \State $p^* \gets$ None
        \For{$(\nrac_i,\krac_i,p_i)\in\mathcal{D}_{\mathrm{RAC}}$}
            \If{$\nrac_i\le a$ and $\krac_i\le b$ and 
                $\mathrm{DP}_{a-\nrac_i,b-\krac_i}$ is not None}
                \State $p \gets \min\{\mathrm{DP}_{a-\nrac_i,b-\krac_i},p_i\}$
                \If{$p^*$ is None or $p>p^*$}
                    \State $(p^*,i^*) \gets (p,i)$
                \EndIf
            \EndIf
        \EndFor
        \State $(\mathrm{DP}_{a,b},C_{a,b}) \gets (p^*,i^*)$
    \EndFor
\EndFor
\State $p_{\mathrm{wc}}^\oplus \gets \mathrm{DP}_{\nrac,\krac}$
\If{$p_{\mathrm{wc}}^\oplus$ is None}
    \State \Return $1/q$, $\emptyset$
\EndIf
\State $(a,b) \gets (\nrac,\krac)$
\State $\mathcal{P} \gets [\,]$
\While{$a\neq0$ or $b\neq0$}
    \State $i \gets C_{a,b}$
    \State $\mathcal{P} \gets \mathcal{P} \mathbin{\|} [\mathcal{D}_{\mathrm{RAC},i}]$
    \State $(a,b) \gets (a-\nrac_i,b-\krac_i)$
\EndWhile
\State \Return $p_{\mathrm{wc}}^\oplus$, $\mathcal{P}$
\end{algorithmic}
\end{algorithm}

\section{Numerical Experiments}
\label{sec:numerics}
We perform three numerical experiments.
First, we compare the proposed OA encoder with the optimal encoder for fixed linear RAC codebooks.
Second, we compare the certified worst-case decoding success probabilities of AG/OA-RACs and RS/OA-RACs for a fixed inner-code database and concatenated-code length.

\subsection{Performance of the OA Encoder}
\label{subsec:distance-numerics}
\newcommand{\ntest}{10}
To evaluate the encoder itself independently of concatenated-code design, we consider binary linear RACs of input length $\nrac=\ntest$.
For each dimension $\krac=1,\ldots,\ntest$, we enumerate linearly inequivalent binary linear codes of length $\ntest$ and identify the largest dual distance $d^\perp_{\max}(\ntest,\krac)$.
Because the OA encoder depends only on the dual distance, its best performance within this ensemble is obtained from a codebook attaining $d^\perp_{\max}(\ntest,\krac)$.
For comparison, we also obtain optimal encoders using the linear programming formulation of \cite{kondo2026random}.
Specifically, we determine an optimal encoder for every inequivalent linear code of length $\ntest$ and, for each code dimension, select the code that maximizes the worst-case decoding success probability.
Note, however, that this optimization is restricted to linear RACs, and the size of the resulting linear program grows exponentially with the parameters.
We also compare with the general Improved RAC construction of~\cite{liabotro2017improved}.
For $\nrac\le2^\krac-1$, its worst-case decoding success probability is
\begin{equation}
    \label{eq:improved-rac}
    \frac12+\frac{1}{2\sqrt{(2^\krac-1)\nrac}}.
\end{equation}
The known upper bound (ANTV U.B.)~\cite{ambainis1999dense} is also included as a reference.

\begin{figure}[t]
    \centering
    \includegraphics[width=.9\columnwidth]{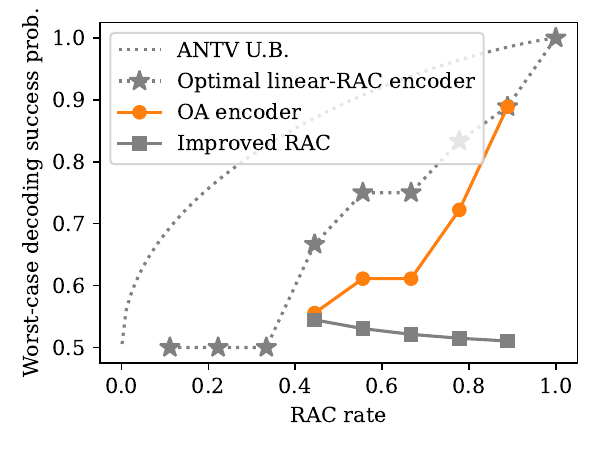}
    \caption{
    Worst-case decoding success probability for binary $(\ntest,\krac)$-RACs. The OA encoder is evaluated using codebooks with the maximum dual distance $d^\perp_{\max}$ for each $\krac$, whereas the optimal encoder is obtained by optimizing over all inequivalent $[\ntest,\krac]$ linear codes.
    }
    \label{fig:finite-rac}
\end{figure}
Figure~\ref{fig:finite-rac} shows that the OA encoder does not attain the optimal encoder performance for most rates, but it substantially improves on the Improved RAC over much of the range in which the latter is applicable.
Closing the gap between the optimal encoder and the OA-based encoder is left for future work.

\subsection{Certified Performance of Proposed RACs with Finite-Length Codes}
\label{subsec:concat-distance-numerics}
\algrenewcommand\algorithmicrequire{\textbf{Input:}}
\algrenewcommand\algorithmicensure{\textbf{Output:}}
\begin{algorithm}[t]
\caption{Parameter selection for an RS concatenated code.}
\label{alg:feasible_parameters_rs}
\begin{algorithmic}[1]
\Require $\ntarget,\ktarget,\mathcal D$
\Ensure Best parameter set $(\nout,\kout,\cin,\dconlb)$, or $\varnothing$ if no certified minimum-distance lower bound of at least $2$ is found
\State $d_{\max}\gets0$
\State $\mathcal F\gets\emptyset$
\For{$\cin\in\mathcal D$}
    \State $\cin$ is an $[\nin,\kin,\din]_\qin$ linear code
    \State $\qout\gets\qin^{\kin}$ \Comment{Concatenation constraint}
    \State $\nout\gets\lceil\ntarget/\nin\rceil$
    \If{$\nout\le\qout$} \Comment{Ordinary RS length constraint}
        \State $\kout\gets\left\lceil\nout(\ktarget/\ntarget)/(\kin/\nin)\right\rceil$
        \If{$1\le\kout\le\nout$}
            \State $\doutlb\gets\nout+1-\kout$
            \State $d_{\mathrm{loss}}\gets\max\{0,\ncon-\ntarget-\kcon+\ktarget\}$
            \State $\dconlb\gets\max\{0,\doutlb\din-d_{\mathrm{loss}}\}$
            \If{$\dconlb\ge 2 \land \dconlb>d_{\max}$}
                \State $d_{\max}\gets \dconlb$
                \State $\mathcal F\gets(\nout,\kout,\cin,\dconlb)$
            \EndIf
        \EndIf
    \EndIf
\EndFor
\If{$\mathcal F=\emptyset$}
    \State \Return $\varnothing$
\Else
    \State \Return $\mathcal F$
\EndIf
\end{algorithmic}
\end{algorithm}

\begin{algorithm}[t]
\caption{Parameter selection for a GS-tower AG concatenated code.}
\label{alg:feasible_parameters_ag}
\begin{algorithmic}[1]
\Require $\ntarget,\ktarget,\mathcal D$
\Ensure Best parameter set $(\kout,m,\cin,\dconlb)$, or $\varnothing$ if no positive guarantee is found
\State $d_{\max}\gets0$
\State $\mathcal F\gets\emptyset$
\For{$\cin\in\mathcal D$}
    \State $\cin$ is an $[\nin,\kin,\din]_\qin$ linear code
    \If{$\kin$ is even or $\qin$ is square} \Comment{Square-field constraint for the GS tower}
        \State $\qout\gets\qin^{\kin}$
        \State $m\gets\max\!\left\{0,\left\lceil2\log_{\qout}\!\left(\frac{\ntarget}{\nin(\qout-\sqrt{\qout})}\right)\right\rceil\right\}$
        \State $\nout\gets\qout^{m/2}(\qout-\sqrt{\qout})$
        \State $\kout\gets\left\lceil\nout(\ktarget/\ntarget)/(\kin/\nin)\right\rceil$
        \State $g\gets$ the value in~\eqref{eq:genus}
        \If{$1\le\kout\le\nout-g$}
            \State $\doutlb\gets\nout+1-g-\kout$
            \State $d_{\mathrm{loss}}\gets\max\{0,\ncon-\ntarget-\kcon+\ktarget\}$
            \State $\dconlb\gets\max\{0,\doutlb\din-d_{\mathrm{loss}}\}$
            \If{$\dconlb>d_{\max}$}
                \State $d_{\max}\gets \dconlb$
                \State $\mathcal F\gets(\kout,m,\cin,\dconlb)$
            \EndIf
        \EndIf
    \EndIf
\EndFor
\If{$\mathcal F=\emptyset$}
    \State \Return $\varnothing$
\Else
    \State \Return $\mathcal F$
\EndIf
\end{algorithmic}
\end{algorithm}

\newcommand{\ntargeta}{200}
\newcommand{\ntargetb}{10^{10}}

In Section~\ref{eq:prac-asymp-racs}, we derived upper and lower bounds on the asymptotic performance of AG/OA-RACs under a constraint on the inner-code length.
Here, we evaluate the finite-length performance of AG/OA-RACs using a database of linear codes with fixed concatenated code length.
For comparison, we also evaluate the performance of RS/OA-RACs.
The parameters of the outer and inner codes are selected according to Algorithms~\ref{alg:feasible_parameters_rs} and~\ref{alg:feasible_parameters_ag}.
For the inner-code database $\mathcal{D}$, we use linear codes of length $\nin\le 30$ from SageMath's~\cite{sagemath} \texttt{best\_linear\_code\_in\_guava} routine, together with their shortened and punctured variants using~\eqref{eq:shorten-puncture-distance}.
If $\mathcal F$ remains empty, none of the candidate inner/outer-code combinations considered by the corresponding algorithm yields a strictly positive certified minimum-distance lower bound after adjustment to the requested $(\nrac,\krac)$.
In this case, we set the certified relative minimum-distance guarantee to zero.
Note that in this experiment, we compute only the certified relative minimum distance of each resulting linear code and the corresponding worst-case decoding success probability of the RAC derived from this guarantee; the generator matrices themselves are not explicitly constructed.
As reference benchmarks, we also include the GV and Zyablov bounds.
The GV bound is a classical existence bound for linear codes, whereas the Zyablov bound gives an existence guarantee for concatenated codes.

\begin{figure}[tb]
    \centering
    \subfloat[]{
        \includegraphics[width=.8\columnwidth]{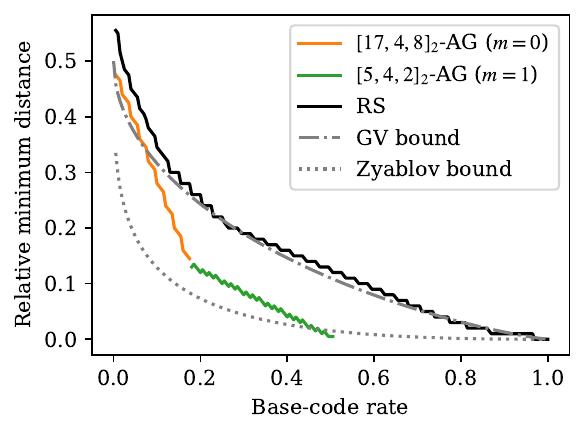}
        \label{fig:concat-a}
    }\\
    \subfloat[]{
        \includegraphics[width=.8\columnwidth]{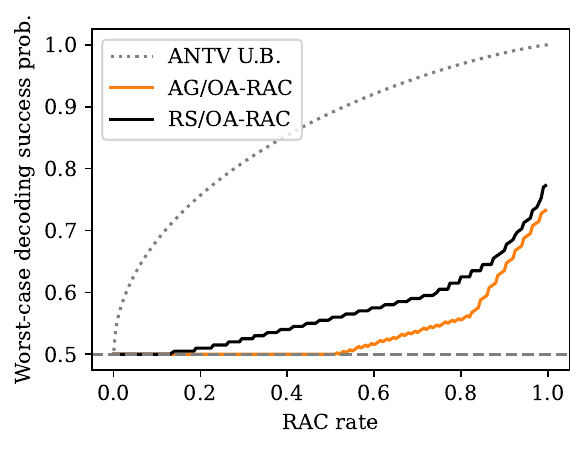}
        \label{fig:rac-a}
    }
    \caption{
    Certified relative minimum-distance lower bounds and worst-case decoding success probabilities for RS and GS-tower AG concatenated codes.
    (a) Certified relative minimum-distance lower bounds as functions of the base-code rate, together with the GV and Zyablov bounds.
    (b) Worst-case decoding success probability as a function of the RAC rate for AG/OA-RAC and RS/OA-RAC, with the asymptotic bounds shown by the shaded region.
    For the AG constructions, the legends indicate the selected inner code and GS-tower level $m$.
    For the RS construction at $n_{\mathrm{con}}=200$, the selected inner code varies across many rate intervals, so the individual inner-code parameters are omitted from the legend for readability.
    }
    \label{fig:concat}
\end{figure}

\begin{figure}[tb]
    \centering
    \includegraphics[width=.9\columnwidth]{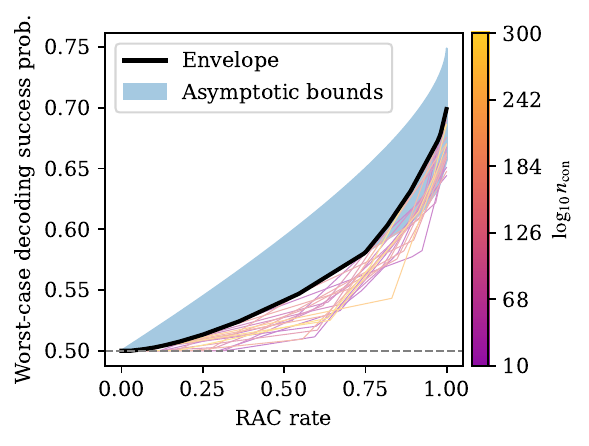}
    \caption{
    Simultaneous plot of the guaranteed worst-case decoding success probabilities for concatenated-code lengths $n_{\mathrm{con}}$ ranging up to $10^{300}$.
    The asymptotic achievable region for AG/OA-RACs derived in Section~\ref{eq:prac-asymp-racs} is also shown.
    Different line colors correspond to different concatenated-code lengths.
    }
    \label{fig:assymp-rac-simul}
\end{figure}

Figures~\ref{fig:concat-a} and \ref{fig:rac-a} show that, at $\nrac=\ntargeta$, the RS construction provides a larger certified relative minimum distance over much of the rate range and approaches the GV benchmark relatively closely.
Thus, at moderate blocklengths, RS/OA-RAC provides a better finite-length worst-case decoding success probability.
Note that the observed oscillations arise from the code-size modification and are not artifacts of the implementation.

As shown in Fig.~\ref{fig:assymp-rac-simul}, for very large concatenated-code lengths, ranging to $10^{300}$, the envelope of the finite-length curves approaches the asymptotic bounds derived in Section~\ref{eq:prac-asymp-racs}, although the individual curves do not.
This distinction arises because each individual curve is evaluated at a fixed concatenated-code length, whereas taking the envelope over different values of $\nrac$ effectively removes this constraint and thus recovers the setting underlying the asymptotic bounds.
Since the total code length is naturally fixed when evaluating a RAC at finite length, this distinction should be kept in mind when interpreting the asymptotic performance of the proposed construction.

\subsection{Practical Performance of DS-RAC}
\label{subsec:}

\begin{figure}[tb]
    \centering
    \subfloat[$\qin=2$]{
        \includegraphics[width=.8\columnwidth]{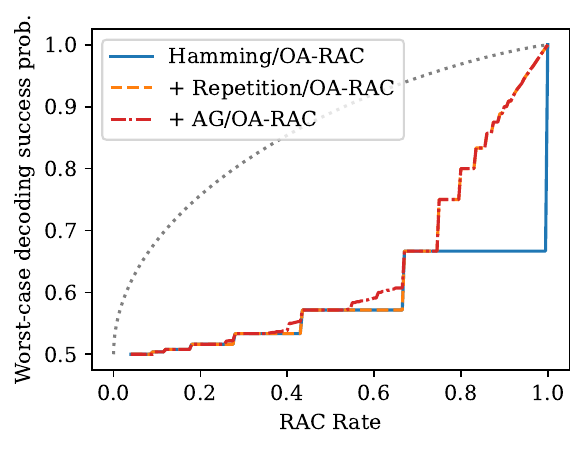}
        \label{fig:directsum-a}
    }\\
    \subfloat[$\qin=5$]{
        \includegraphics[width=.8\columnwidth]{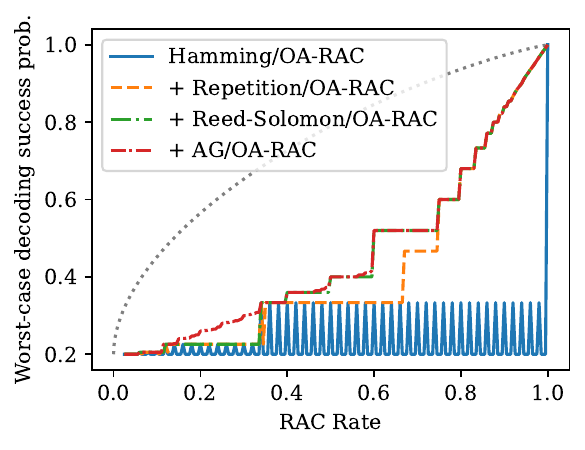}
        \label{fig:directsum-b}
    }
    \caption{
    Worst-case decoding success probability of $(200,\krac)$-RACs constructed using the direct-sum strategy for (a) $\qin=2$ and (b) $\qin=5$.
    The gray dotted curve represents the ANTV upper bound~\cite{ambainis1999dense}.
    }
    \label{fig:directsum}
\end{figure}

We numerically evaluated the worst-case decoding success probability of DS-RACs.
We fixed the input length to $\nrac=200$ and considered message lengths $\krac=1,2,\ldots,\nrac$.
The RAC dataset $\mathcal{D}_{\mathrm{RAC}}$ used in the DS-RAC consisted of the following:
\begin{itemize}
    \item RACs constructed using OA encoders with Hamming codes, repetition codes, or RS codes as the base codes, with $\nrac\in\{1,2,\ldots,200\}$ and $\krac\in\{1,2,\ldots,\nrac\}$.
    For $\qin=2$, RS codes were excluded because they do not yield nontrivial RACs.
    \item $(\nrac,\krac)_\qin$-AG/OA-RACs with $\nrac\in\{1,2,\ldots,200\}$ and $\krac\in\{1,2,\ldots,\nrac\}$ whose worst-case decoding success probability strictly exceeds $1/q$.
\end{itemize}

Figure~\ref{fig:directsum} shows that DS-RACs with Hamming/OA-RACs alone achieve a worst-case decoding success probability greater than $1/q$ over a wide range of RAC rates.
For $q=5$, the worst-case decoding success probability does not exceed $1/q$ at some rates.
This does not contradict Theorem~\ref{theorem:directsum}, since the theorem provides a guarantee only in the asymptotic limit $\nrac\to\infty$.

The figure also shows that successively incorporating RACs from other infinite families improves the worst-case decoding success probability.
For $\qin=2$, adding Repetition/OA-RACs to Hamming/OA-RACs already yields a substantial improvement at high RAC rates, whereas the additional improvement obtained by including AG/OA-RACs is relatively small around an RAC rate of $0.6$.
A similar trend is observed for $\qin=5$.
In this case, however, Reed--Solomon/OA-RACs substantially improve the worst-case decoding success probability for RAC rates of approximately $0.3$--$0.7$, while AG/OA-RACs provide a substantial improvement for rates of approximately $0.2$--$0.3$.

We note that the performance of the AG/OA-RACs was evaluated using a lower bound on the minimum distance.
Therefore, a tighter lower bound, i.e., a larger value of $d^\perp_{\mathrm{lb}}$ supplied to the OA encoder, is expected to yield an even higher worst-case decoding success probability.

\section{Conclusion}
\label{sec:conclusion}
In this work, we introduced the notion of asymptotically good RACs,
whose compression fractions and worst-case decoding advantages over random
guessing both remain bounded away from zero.

As the first approach to constructing such RACs, we introduced an OA encoder
based on the orthogonal-array property of linear codes.
Using dual codes as RAC codebooks, we showed that asymptotically good linear
codes yield asymptotically good OA-RACs.
We further gave an explicit polynomial-time realization using concatenated
codes with Garcia--Stichtenoth algebraic-geometry outer codes.

As the second approach, we proposed a direct-sum strategy that constructs large RACs, called DS-RACs, by combining smaller RACs.
We showed that, with a dataset containing all $q$-ary Hamming/OA-RACs and the trivial $(1,1)_q$-RAC, this strategy constructs asymptotically good RAC families at any prescribed rate $R\in(0,1)$.

Our numerical results demonstrated the finite-length performance of both approaches.
OA-RACs substantially outperform the Improved RAC construction~\cite{liabotro2017improved} over much of the examined range, although they do not generally attain the optimum.
For the direct-sum strategy, incorporating OA-RACs derived from additional linear-code families, including algebraic-geometry codes, further improves the worst-case decoding success probability over a wide range of RAC rates.
Developing improved encoders, constructing explicit RAC families closer to the entropy-optimal rate--success tradeoff, and reducing the sampling complexity below the present cubic-time implementation are natural directions for future work.

\section{ACKNOWLEDGMENT}
\label{sec:acknowledgement}
The authors used ChatGPT (GPT-5.6 Sol, OpenAI) for brainstorming, checking mathematical expressions and derivations in Sections \ref{sec:preliminaries}--\ref{sec:numerics}, and assistance with Python code development.
All AI-assisted content and code were independently reviewed and verified by the authors, who take full responsibility for the final manuscript.

\bibliographystyle{unsrtnat}
\bibliography{reference}

\end{document}